\documentclass[conference]{IEEEtran}
\usepackage[cmex10]{amsmath}
\usepackage{amssymb}
\usepackage{graphicx}
\usepackage{mathptmx}

\DeclareMathAlphabet{\mathcal}{OMS}{cmsy}{m}{n}

\allowdisplaybreaks[0]

\newtheorem{theorem}{Theorem}

\newtheorem{lemma}[theorem]{Lemma}
\newtheorem{definition}[theorem]{Definition}
\newtheorem{corollary}[theorem]{Corollary}

\newcommand{\sasoglu}{{\c S}a{\c s}o{\u g}lu}

\begin{document}
\title{Non-Binary Polar Codes using Reed-Solomon Codes and Algebraic Geometry Codes}


\author{
\IEEEauthorblockN{Ryuhei Mori and Toshiyuki Tanaka}
\IEEEauthorblockA{Graduate School of Informatics\\
Kyoto University \\
Kyoto, 606--8501, Japan\\
Email: rmori@sys.i.kyoto-u.ac.jp, tt@i.kyoto-u.ac.jp}
}

\maketitle

\begin{abstract}
Polar codes, introduced by Ar{\i}kan, achieve symmetric capacity of any discrete memoryless channels
under low encoding and decoding complexity.
Recently, non-binary polar codes have been investigated.
In this paper, we calculate error probability of non-binary polar codes constructed on the basis of
Reed-Solomon matrices by numerical simulations.
It is confirmed that 4-ary polar codes have significantly better performance than binary polar codes
on binary-input AWGN channel.
We also discuss an interpretation of polar codes in terms of algebraic geometry codes, 
and further show that polar codes using Hermitian codes have asymptotically good performance.
\end{abstract}
\section{Introduction}
Ar{\i}kan~\cite{5075875} proposed polar codes as codes that achieve 
symmetric capacity of arbitrary binary-input discrete memoryless channels (DMCs)
under low-complexity encoding and decoding. 
Asymptotic error probability of polar codes has been studied in detail by 
Ar{\i}kan and Telatar~\cite{arikan2008rcp}, who showed 
that the error probability is $o(2^{-N^\beta})$ for any $\beta<1/2$
and $\omega(2^{-N^\beta})$ for any $\beta>1/2$, where $N$ is the blocklength.
Although polar codes are constructed on the basis of 
a Kronecker power of the matrix $\begin{bmatrix}1&0\\1&1\end{bmatrix}$ 
in Ar\i kan's original proposal, 
one can extend polar codes using a different matrix.  
Korada, \sasoglu, and Urbanke discussed such extensions 
to improve the threshold of $\beta$, which equals 1/2 
when polar codes are constructed using the $2\times 2$ matrix mentioned above, 
and showed that the threshold can indeed be made larger 
by using a larger matrix~\cite{korada2009pcc}.

Another important direction of extending polar codes 
is to consider non-binary input alphabet.  
\sasoglu, Telatar, and Ar{\i}kan considered non-binary polar codes and showed that
polar codes achieve symmetric capacity when the size of input alphabet is a prime~\cite{sasoglu2009pad}.
They also showed that one can still obtain capacity-achieving codes 
even when the size of input alphabet is not a prime, 
by decomposing the original channel to multiple channels, each of which 
has input alphabet whose cardinality is a prime, and
by using a polar code for each of these channels~\cite{5351487}.
This method of decomposition is also known as multilevel coding~\cite{1055718}.

In~\cite{mori2010cpq}, the authors discussed the case in which the size of input alphabet is an integer power of a prime, 
and showed that polar codes defined on the input alphabet can achieve,
without the decomposition,
symmetric capacity and that use of a larger matrix can improve 
the asymptotic error probability, similarly to the binary-input case.
Furthermore, it is shown
that Reed-Solomon matrices can be regarded as a natural generalization of the binary $2\times 2$ matrix, 
providing a family of polar codes with various nice properties.  

In this paper, we calculate error probability of non-binary polar codes
using Reed-Solomon matrices on the $q$-ary erasure channels by numerical simulations.
We further show that polar codes using Hermitian codes have asymptotically good performance.

\section{Non-Binary Polar Codes}
Assume that $q$ is an integer power of a prime.
Let $W:\mathbb{F}_q\to\mathcal{Y}$ be a $q$-ary DMC and $G$ be an $\ell\times\ell$ matrix on $\mathbb{F}_q$.
An $\ell^n\times \ell^n$ matrix $G_{\ell^n}$ is defined as
$G_{\ell^n} := G^{\otimes n}$, where ${}^\otimes$ is the Kronecker power.
Let $u_0^{\ell^n-1}$ be a row vector $(u_0,\dotsc,u_{\ell^n-1})$ and 
$u_i^j$ be its subvector $(u_i,\dotsc,u_j)$.
Let $u_\mathcal{F}$ be a subvector $(u_{f_0},\dotsc,u_{f_{m-1}})$ of $u_0^{\ell^n-1}$
where $\mathcal{F} = \{f_0,\dotsc,f_{m-1}\} \subseteq \{0,\dotsc,\ell^n-1\}$.
Let $\mathcal{F}^c$ be the complement of $\mathcal{F}$.
Let $W^{\ell^n}:\mathbb{F}_q^{\ell^n}\to\mathcal{Y}^{\ell^n}$ denote the DMC defined as
$W^{\ell^n}(y_0^{{\ell^n}-1}\mid u_0^{{\ell^n}-1}) := \prod_{i=0}^{{\ell^n}-1} W(y_i\mid u_i)$.
The $\ell$-ary bit-reversal matrix $R_{\ell^n}$ of size $\ell^n$ is a permutation matrix
defined by $u_0^{\ell^n-1} R_{\ell^n} = (u_{r_0},\dotsc,u_{r_{\ell^n-1}})$
where $\ell$-ary expansion $a_1\dotsm a_n$ of $i$ and $\ell$-ary expansion $a_n\dotsm a_1$ of $r_i$ 
are the reversals of each other.
The encoder of polar codes is defined as $\phi(u_{\mathcal{F}^c}) := u_0^{\ell^n-1} R_{\ell^n} G_{\ell^n}$
where the all-zero vector is assigned to $u_\mathcal{F}$.
The matrix $G$ is called a kernel of the polar code.
The decoder of polar codes is a successive cancellation (SC) decoder.
For $i\in\{0,\dotsc,\ell^n-1\}$, $\hat{u}_0^{i-1}\in\mathbb{F}_q^i$ and $y_0^{\ell^n-1}\in\mathcal{Y}^{\ell^n}$,
let
\begin{equation*}
\psi_i\left(\hat{u}_0^{i-1}, y_0^{\ell^n-1}\right) = \mathop{\rm arg\,max}_{u_i\in\mathbb{F}_q}
P_{U_i\mid U_0^{i-1}, Y_0^{\ell^n-1}}\left(u_i\bigm| \hat{u}_0^{i-1}, y_0^{\ell^n-1}\right)
\end{equation*}
where $U_0^{\ell^n-1}$ and $Y_0^{\ell^n-1}$ are random variables which obey the distribution
\begin{multline*}
P_{U_0^{\ell^n-1}, Y_0^{\ell^n-1}}\left(u_0^{\ell^n-1}, y_0^{\ell^n-1}\right)\\
=\frac1{q^{\ell^n}}W^{\ell^n}\left(y_0^{\ell^n-1}\bigm| u_0^{\ell^n-1}R_{\ell^n}G_{\ell^n}\right).
\end{multline*}
An output $\hat{u}_0^{\ell^n-1}$ of the decoder is determined sequentially from $\hat{u}_0$ to $\hat{u}_{\ell^n-1}$ as
\begin{equation*}
\hat{u}_i = \begin{cases}
0, & \text{if } i \in \mathcal{F}\\
\psi_i\left(\hat{u}_0^{i-1}, y_0^{\ell^n-1}\right), & \text{otherwise.}
\end{cases}
\end{equation*}

Let
\begin{equation*}
P_{\ell^n}^{(i)} := P\left(\psi_i(U_0^{i-1}, Y_0^{\ell^n-1}) \ne U_i\right).
\end{equation*}
In order to obtain polar codes of small error probability, $\mathcal{F}^c$ have to be chosen such that
$P_{\ell^n}^{(i)}$ is small if $i\in\mathcal{F}^c$.
The error probabilities $\{P_{\ell^n}^{(i)}\}$ can be calculated by using density evolution~\cite{5205857}.

\section{Exponent of Polar Codes and Reed-Solomon kernel}
\subsection{Exponent of polar codes}
Asymptotic performance of polar codes is determined by a kernel $G$.
Ar{\i}kan and Telatar showed that when $q=2$ and $G=\begin{bmatrix}1&0\\1&1\end{bmatrix}$,
the error probability of polar codes is $o(2^{-2^{\beta n}})$ for any $\beta<1/2$
and $\omega(2^{-2^{\beta n}})$ for any $\beta>1/2$~\cite{arikan2008rcp}.
Korada, \sasoglu, and Urbanke~\cite{korada2009pcc} generalized the result to any $G$ when $q=2$, 
showing that there exists a function $E(G)\in[0,1)$ such that
the error probability of polar codes is $o(2^{-\ell^{\beta n}})$ for any $\beta<E(G)$
and $\omega(2^{-\ell^{\beta n}})$ for any $\beta>E(G)$.
The threshold $E(G)$ of $\beta$ is called the exponent of a kernel $G$.
They also showed that $\max_{G\in\mathbb{F}_2^{\ell\times\ell}} E(G)$ converges to 1 as $\ell\to\infty$.
They further proposed an explicit construction method of $G$ using BCH codes, 
for which the exponent $E(G)$ can be made arbitrarily close to 1 as $\ell$ becomes large~\cite{korada2009thesis}.

In~\cite{mori2010cpq}, the authors showed that the result about exponent can further be generalized to $q$-ary polar codes.
The exponent $E(G)$ of a kernel $G$ can easily be calculated for $q$ being equal to an integer power of a prime, as follows.
\begin{theorem}[\cite{korada2009pcc}]\label{thm:eg}
\begin{equation*}
E(G) = \frac1\ell \sum_{i=0}^{\ell-1} \log_\ell D_i
\end{equation*}
where the partial distance $D_i$ is defined as
\begin{equation*}
D_i := \min_{v_{i+1}^{\ell-1}\in\mathbb{F}_q^{\ell-i-1}} d\left((0_0^{i-1}, 0, v_{i+1}^{\ell-1})G, (0_0^{i-1}, 1, 0_{i+1}^{\ell-1})G\right).
\end{equation*}
In the above definition of $D_i$, $d(x,y)$ denotes the Hamming distance between $x$ and $y$.
\end{theorem}
\begin{definition}
$L(q, \ell):=\max_{G\in\mathbb{F}_q^{\ell\times\ell}}E(G)$.
\end{definition}
As in the case $q=2$~\cite{korada2009pcc}, 
the best exponent $L(q,\ell)$ can be lower bounded by using the Gilbert-Varshamov-like bound.
\begin{lemma}
\begin{equation*}
L(q, \ell) \ge \frac1\ell \sum_{i=0}^{\ell-1} \log_\ell \tilde{D}_i
\end{equation*}
where
\begin{equation*}
\tilde{D}_i := \max\left\{D\in\mathbb{N}\mid \sum_{j=0}^{D-1}\binom{\ell}{j} (q-1)^j < q^{i+1}\right\}.
\end{equation*}
\end{lemma}
\begin{corollary}
\begin{equation*}
\lim_{\ell\to\infty} L(q, \ell) = 1
\end{equation*}
\end{corollary}
\begin{proof}
Let
\begin{equation*}
\omega(\alpha) := \lim_{\ell\to\infty}\frac{\tilde{D}_{\lceil \alpha\ell \rceil}}{\ell}
\end{equation*}
for $\alpha\in[0,\,1]$.  
Then, the equality 
\begin{equation*}
h(\omega(\alpha)) + \log_2(q-1) \omega(\alpha) = \alpha\log_2 q
\end{equation*}
holds for any $0\le\omega(\alpha)\le1/2$, where $h(\cdot)$ is the binary entropy function.  
Hence, $\omega(\alpha)=0$ if and only if $\alpha=0$.
The best exponent $L(q, \ell)$ is bounded from below as 
\begin{align*}
L(q, \ell) &\ge \frac1\ell \sum_{i=0}^{\ell-1} \log_\ell \tilde{D}_i\\
&\ge \frac1\ell \sum_{i=\lceil\alpha \ell\rceil}^{\ell-1} \log_\ell \tilde{D}_i\\
&\ge \frac1\ell (1-\alpha)\ell \log_\ell \tilde{D}_{\lceil\alpha \ell\rceil}\\
&= (1-\alpha) \left(1+\log_\ell\left(\tilde{D}_{\lceil\alpha \ell\rceil}/\ell\right)\right),
\end{align*}
where $\tilde{D}_{\lceil\alpha \ell\rceil}/\ell$ approaches a nonzero limit $\omega(\alpha)$ as $\ell\to\infty$ for any fixed $0<\alpha\le1$. 
Hence, 
$\liminf_{\ell\to\infty} L(q,\ell) \ge 1-\alpha$
for any fixed $0<\alpha\le1$.
\end{proof}

\subsection{Reed-Solomon kernel}
Generator matrices of Reed-Solomon codes are considered suitable as a kernel of polar codes
since they have the following two properties:
(1) Low-rate Reed-Solomon codes are subcodes of Reed-Solomon codes with higher rates.
(2) Minimum distance of Reed-Solomon codes coincides with the Singleton bound.
From these properties,
for any $\ell\in\{2,\dots,q\}$,
one can obtain the $q$-ary matrix
$G_{\text{RS}}(q,\ell)$ of size $\ell\times\ell$ whose submatrix consisting of $i$-th row to $(\ell-1)$-th row
is a generator matrix of $[\ell, \ell-i, i+1]_q$ Reed-Solomon code.
We call the $q$-ary matrix $G_\text{RS}(q,\ell)$ the Reed-Solomon matrix.
From Theorem~\ref{thm:eg}, $E(G_\text{RS}(q,\ell))=\log(\ell!)/(\ell\log \ell)$.
The second property mentioned above guarantees the optimality of the Reed-Solomon matrix $G_\text{RS}(q,\ell)$ 
as a kernel of polar codes, 
yielding $L(q,\ell) = \log(\ell!)/(\ell\log\ell)$ for all $\ell\le q$.
One obtains, for example, $L(4,4) \approx 0.573\,12$,
while in the binary case $L(2,31) \lesssim 0.55$ and $L(2,16) \approx 0.518\,28$~\cite{korada2009pcc}. 
This example demonstrates efficiency of using a non-binary kernel in constructing polar codes even for a binary-input DMC.

The original binary $2\times 2$ matrix can be identified as $G_\text{RS}(2,2)$.
This implies that $G_\text{RS}(q,q)$ can be regarded as a natural generalization of the binary $2\times 2$ matrix.
In Subsection \ref{subsec:as},
we see relation between polar codes using $G_\text{RS}(q,q)$ and $q$-ary Reed-Muller codes, 
which has been mentioned by Ar{\i}kan in the binary ($q=2$) case~\cite{5075875}.

\section{Numerical Simulation Results on the $q$-ary Erasure Channel}
\subsection{Recursive calculation of error probability on the $q$-ary erasure channel}
In this section, we evaluate performance of polar codes constructed on the basis of 
the Reed-Solomon matrices. 
We consider the $q$-ary erasure channel for simplicity.
For $\epsilon\in(0,1)$, the $q$-ary erasure channel $W:\mathbb{F}_q\to\mathbb{F}_q\cup\{*\}$ is defined as
\begin{align*}
W(y\mid x) := \begin{cases}
\epsilon, &\text{if } y=*\\
1-\epsilon, &\text{if } y=x\\
0, &\text{otherwise}
\end{cases}
\end{align*}
for any $x\in\mathbb{F}_q$.
When a Reed-Solomon matrix is used as a kernel of polar codes, $P_{\ell^n}^{(i)}$ can be calculated recursively.
Since the Reed-Solomon codes are maximum distance separable (MDS) codes,
Reed-Solomon codes are correctable if and only if the number of erased symbols is smaller than
the minimum distance.
From this observation, one obtains the recursion formula 
\begin{equation*}
P_{\ell^n}^{(a\ell + b)} = \sum_{i=b+1}^\ell \binom{\ell}{i} {P_{\ell^{n-1}}^{(a)}}^i \left(1-P_{\ell^{n-1}}^{(a)}\right)^{\ell-i}
\end{equation*}
for $0\le a\le \ell^{n-1}-1$ and $0\le b\le \ell-1$.

\begin{figure}[p]
\includegraphics[width=\hsize]{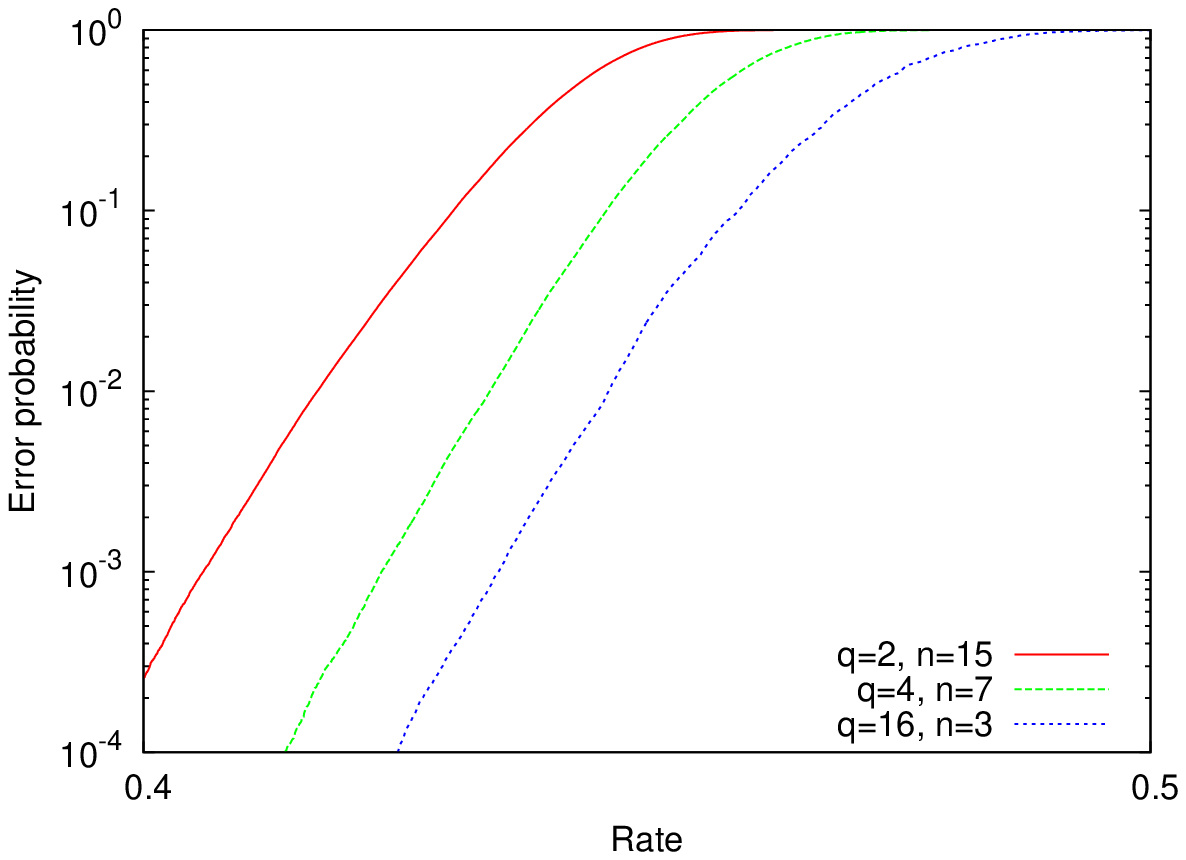}
\caption{Performance comparison of $q$-ary polar codes on $G_\text{RS}(q,q)$ on the $q$-ary erasure channels.
Blocklengths of the binary and 4-ary codes are $2^{15}$ viewed as binary codes.
Blocklength of the 16-ary code is $2^{14}$ viewed as a binary code.
Erasure probabilities of all channels are 0.5.
}
\label{fig:polar-compare}

\includegraphics[width=\hsize]{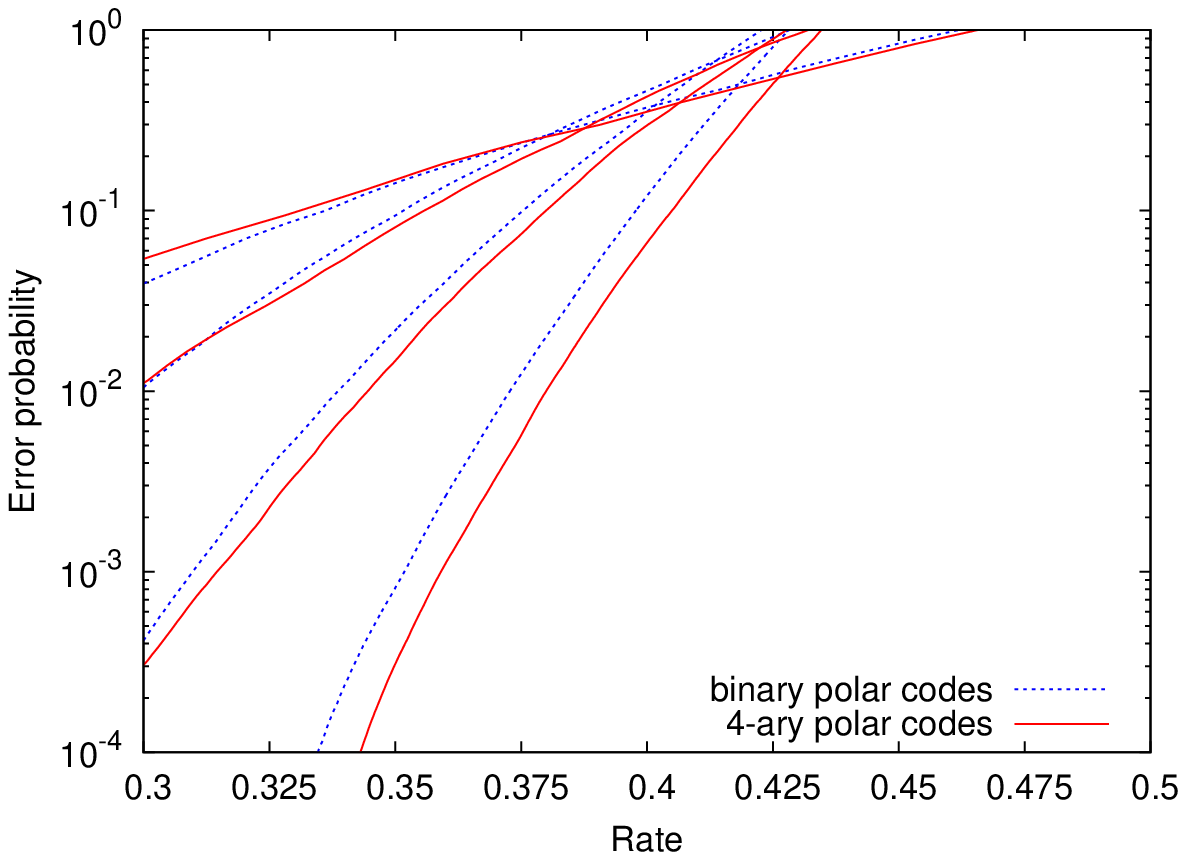}
\caption{
Performance comparison of binary polar codes on $G_\text{RS}(2,2)$ and 4-ary polar codes on $G_\text{RS}(4,2)$
on binary-input AWGN channel.
Blocklengths are $2^7$, $2^9$, $2^{11}$, and $2^{13}$ viewed as binary codes.
The results for 4-ary polar codes and binary polar codes
are plotted by solid curves and dotted curves, respectively.
}
\label{fig:awgn}

\includegraphics[width=\hsize]{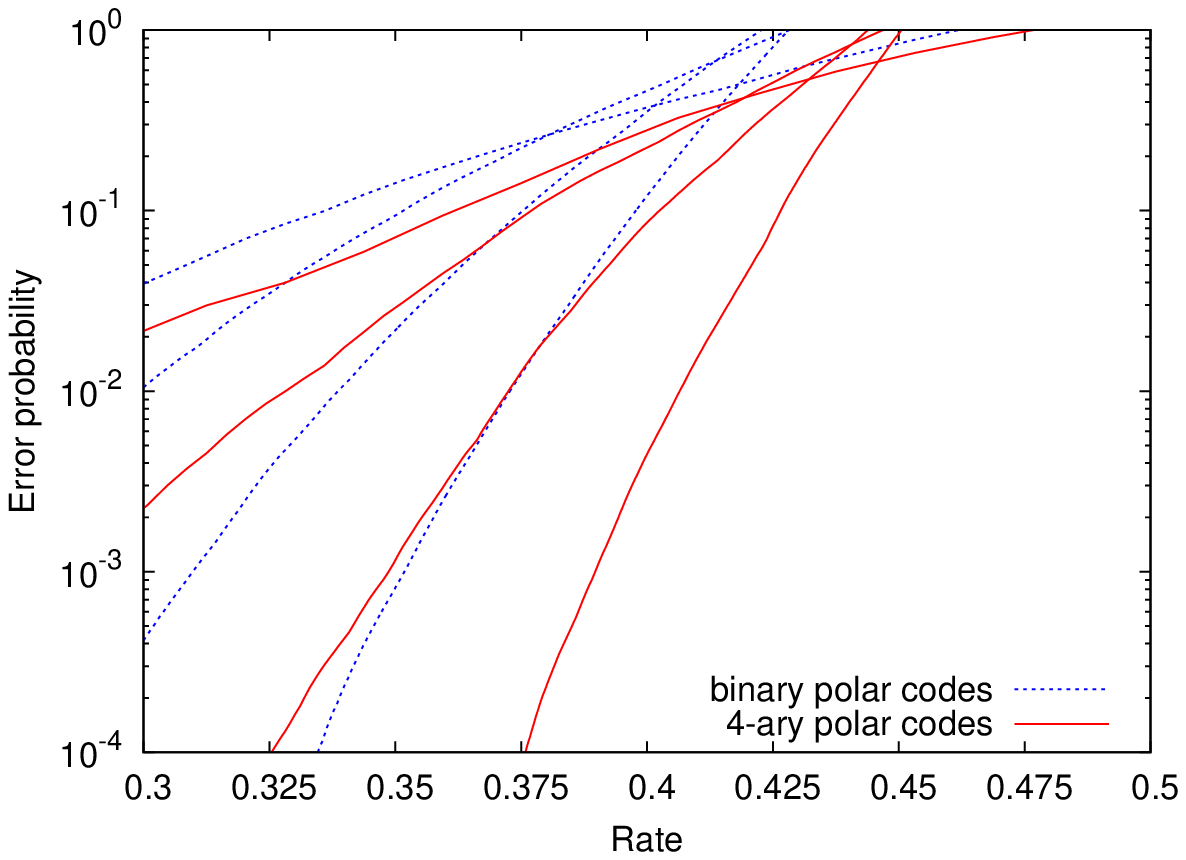}
\caption{
Performance comparison of binary polar codes on $G_\text{RS}(2,2)$ and 4-ary polar codes on $G_\text{RS}(4,4)$
on binary-input AWGN channel.
Blocklengths are $2^7$, $2^9$, $2^{11}$, and $2^{13}$ viewed as binary codes.
The results for 4-ary polar codes and binary polar codes
are plotted by solid curves and dotted curves, respectively.
}
\label{fig:4RS-awgn}
\end{figure}

\subsection{Numerical simulation results}
In this subsection, simulation results of $2^m$-ary polar codes are shown.
The blocklength of $2^m$-ary polar codes is $m\ell^{n}$ viewed as binary codes where $\ell$ is the size of a kernel $G$
and a submatrix of $G^{\otimes n}$ is used for generator matrix of polar codes.

Error probabilities of binary, 4-ary, and 16-ary polar codes using $G_\text{RS}(q,q)$ on the $q$-ary erasure channels
are shown in Fig.~\ref{fig:polar-compare}.
Blocklengths of the binary and 4-ary polar codes are $2^{15}$ viewed as binary codes.
Blocklength of the 16-ary polar code is $2^{14}$ viewed as a binary code.
These results imply that error probability of polar codes using $G_\text{RS}(q,q)$ for a large $q$ is also small in practical blocklength,
although it should be noted that in Fig.~\ref{fig:polar-compare} 
the binary, 4-ary, and 16-ary polar codes are simulated on different channels, 
namely, the binary, 4-ary, and 16-ary erasure channels, respectively.

Blockwise independent $q$-ary channels of blocksize $\ell$ can be 
viewed as a single $q^\ell$-ary memoryless channel, 
so that one can achieve capacity with $q^\ell$-ary polar codes.  
One can alternatively use $q$-ary polar codes with a kernel of size 
multiple of $\ell$, and still provably achieve capacity.
More precisely, a kernel of size multiple of $\ell$ is required only in the first channel transform~\cite{5351487}, \cite{1055718}.
Two binary subchannels constructed from 
the 4-ary erasure channel of erasure probability $\epsilon$ 
by the channel transform with $G_\text{RS}(2,2)$
are both the binary erasure channels of erasure probability $\epsilon$.
Hence, the error probability of the binary polar codes using $G_\text{RS}(2,2)$ on the 4-ary erasure channel is 
equal to the error probability of the binary polar codes using $G_\text{RS}(2,2)$ of the same rate and half blocklength on the binary erasure channel.
Hence, Fig.~\ref{fig:polar-compare} shows that the 4-ary polar codes have significantly better performance than the binary polar codes
on the 4-ary erasure channel.

Simulation results on binary-input AWGN channel are shown in Figs.~\ref{fig:awgn} and \ref{fig:4RS-awgn}.
The standard deviation of noise is 0.978\,65.
The capacity of the binary-input AWGN channel is about 0.5.
In Figs.~\ref{fig:awgn} and \ref{fig:4RS-awgn}, the binary expansion $(\mathbb{F}_4=\{0,1,\alpha,\alpha^2\})\to (\mathbb{F}_2^2=\{00, 01, 10, 11\})$ defined as
$0\to 00, 1\to 01, \alpha\to 10, \alpha^2\to 11$ is used for assigning 4-ary symbols to inputs of the binary-input AWGN channels.
In order to avoid high computational complexity of multi-dimensional density evolution,
$\{P_{\ell^n}^{(i)}\}$ is evaluated by numerical simulation.
The sums $\sum_{i=0}^k \tilde{P}_{\ell^n}^{(i)}$, which are upper bounds of error probabilities are empirically evaluated and plotted,
where $\{\tilde{P}_{\ell^n}^{(i)}\}$ is the sorted version of $\{P_{\ell^n}^{(i)}\}$ 
according to their magnitudes.
The upper bound is considered to be tight if rate is not close to the capacity~\cite{5205857}.
In Fig.~\ref{fig:awgn}, the binary polar codes using $G_\text{RS}(2,2)$ and 4-ary polar codes using $G_\text{RS}(4,2)$ are simulated.
Instead of the standard Reed-Solomon matrix $\begin{bmatrix}1&0\\1&1\end{bmatrix}\in\mathbb{F}_4^4$,
a modified matrix $\begin{bmatrix}1&0\\1&\alpha\end{bmatrix}\in\mathbb{F}_4^4$ is used as $G_\text{RS}(4,2)$
since each bit of the binary image of a 4-ary symbol is independently polarized by $\begin{bmatrix}1&0\\1&1\end{bmatrix}\in\mathbb{F}_4^4$.
In Fig.~\ref{fig:awgn}, there are small differences of performance between the binary and 4-ary polar codes.
In Fig.~\ref{fig:4RS-awgn}, the results of the binary polar codes using $G_\text{RS}(2,2)$ and 4-ary polar codes using $G_\text{RS}(4,4)$ are plotted.
It can be confirmed that 4-ary polar codes using $G_\text{RS}(4,4)$ have significantly better performance than binary polar codes using $G_\text{RS}(2,2)$.
In Figs.~\ref{fig:awgn} and \ref{fig:4RS-awgn}, the blocklengths are $2^7$, $2^9$, $2^{11}$, and $2^{13}$ viewed as binary codes.

\section{Polar Codes and Algebraic Geometry Codes}
\subsection{Polar codes as algebraic geometry codes}\label{subsec:as}
In~\cite{5075875}, Ar{\i}kan mentioned relation between binary polar codes and binary Reed-Muller codes.
The relationship can be naturally generalized to $q$-ary cases.
In this subsection, we overview how $q$-ary Reed-Muller codes are constructed from
$q$-ary Reed-Solomon matrix $G_\text{RS}(q,q)$ using the Kronecker power.
\begin{definition}
Let $q$ be an integer power of a prime.
For any $n\in\mathbb{N}$ and $r=0,\dotsc,(q-1)n$, the $q$-ary $r$-th order Reed-Muller codes are defined as
$\{(p(a_1),\dotsc,p(a_{q^n}))\mid p \in \mathbb{F}_q[X_1,\dotsc,X_n], \deg(p)\le r\}$,
where $\{a_1,\dotsc a_{q^n}\} = \mathbb{F}_q^n$ and where $\deg(p)$ is the degree of a polynomial 
$p \in \mathbb{F}_q[X_1,\dotsc,X_n]$.
\end{definition}
It should be noted that the Reed-Muller codes with $n=1$ are also called extended Reed-Solomon codes.

The binary $2\times 2$ matrix can be regarded as the Reed-Solomon matrix $G_\text{RS}(2,2)$:
\begin{align*}
X:&~\;1~~~ 0\\
\begin{array}{l}
X\\
1
\end{array}
&
\begin{bmatrix}
1 & 0\\
1 & 1
\end{bmatrix}.
\end{align*}

By using the Kronecker product on $G_\text{RS}(2,2)$, a generator matrix of binary 2-variable Reed-Muller codes is obtained as follows.
\begin{align*}
(X_2,X_1):&(1,1)
(1,0)
(0,1)
(0,0)
\\
\begin{array}{l}
X_2X_1\\
X_2\\
\phantom{X_2}X_1\\
\phantom{X_2}1
\end{array}
&
\begin{bmatrix}
1~~& 0~~& 0~~& 0~~\\
1~~& 1~~& 0~~& 0~~\\
1~~& 0~~& 1~~& 0~~\\
1~~& 1~~& 1~~& 1~~
\end{bmatrix}.
\end{align*}
This method of construction of the binary Reed-Muller codes corresponds to the Plotkin construction.
We can see that the binary expansion of $2^n-1-i$ corresponds to
the monomial in the $i$-th row of the Reed-Muller matrix.

Similar relation also holds for non-binary cases.
For example, the Reed-Solomon matrix $G_\text{RS}(3,3)$ is
\begin{align*}
X:&~~2~~~1~~~0\\
\begin{array}{l}
X^2\\
X\\
1
\end{array}
&
\begin{bmatrix}
1 & 1 & 0\\
2 & 1 & 0\\
1 & 1 & 1
\end{bmatrix}.
\end{align*}

A generator matrix of ternary 2-variable Reed-Muller codes is obtained 
by using the Kronecker product on $G_\text{RS}(3,3)$.
\begin{align*}
(X_2,X_1):&(2,2)\,(2,1)\,(2,0)\,(1,2)\,(1,1)\,(1,0)\,(0,2)\,(0,1)\,(0,0)\\
\begin{array}{l}
X_2^2X_1^2\\
X_2^2X_1\\
X_2^2\\
X_2X_1^2\\
X_2X_1\\
X_2\\
\phantom{X_2}X_1^2\\
\phantom{X_2}X_1\\
\phantom{X_2}1
\end{array}
&
\begin{bmatrix}
~1~ & ~1~ & ~0~ & ~1~ & ~1~ & ~0~ & ~0~ & ~0~ & ~0~\\
2 & 1 & 0 & 2 & 1 & 0 & 0 & 0 & 0\\
1 & 1 & 1 & 1 & 1 & 1 & 0 & 0 & 0\\
2 & 2 & 0 & 1 & 1 & 0 & 0 & 0 & 0\\
1 & 2 & 0 & 2 & 1 & 0 & 0 & 0 & 0\\
2 & 2 & 2 & 1 & 1 & 1 & 0 & 0 & 0\\
1 & 1 & 0 & 1 & 1 & 0 & 1 & 1 & 0\\
2 & 1 & 0 & 2 & 1 & 0 & 2 & 1 & 0\\
1 & 1 & 1 & 1 & 1 & 1 & 1 & 1 & 1\\
\end{bmatrix}
\end{align*}
Similarly to the binary case,
the ternary expansion of $3^n-1-i$ corresponds to
the monomial in the $i$-th row of the Reed-Muller matrix.
These observations imply that polar codes using Reed-Solomon matrices can be naturally regarded as codes spanned by polynomials.
A similar property also holds for the case using matrices related with Hermitian codes, as discussed in the next subsection.

Note that the selection rule of rows from $G_\text{RS}(q,q)^{\otimes n}$ for Reed-Muller codes does not maximize the minimum distance unless $q=2$.
In order to maximize the minimum distance, rows have to be chosen according to $\prod_{j=1}^n (i_j+1)$ where $i_j$ is
the $j$-th digit of the $q$-ary expansion of the index $i$ of a row.
In this paper, we call codes based on the rule which maximizes minimum distance hyperbolic codes,
which are also called Massey-Costello-Justesen codes~\cite{massey1973polynomial} and hyperbolic cascaded Reed-Solomon codes~\cite{saints1993hyperbolic}.
On fixed positive rate, the minimum distance of Reed-Muller codes is $q^{\frac12n+o(n)}$~\cite{1054127} while the minimum distance of
polar codes and hyperbolic codes are $q^{E(G_\text{RS}(q,q))n+o(n)}$.
Hence, Reed-Muller codes have asymptotically worse performance than polar codes in non-binary case.

\subsection{Polar codes on algebraic geometry codes}
In order to obtain large exponents, 
algebraic geometry codes are considered suitable as kernels of polar codes, since they can have large minimum distance
and often have the same nested structure as Reed-Solomon codes.
In order to demonstrate feasibility of constructing polar codes using algebraic geometry codes,
we use Hermitian codes as an example.
\begin{definition}
Let $r$ be a power of a prime and $q=r^2$.
A function $\rho:\mathbb{F}_q[X_1,X_2]\to\mathbb{N}\cup\{0,-\infty\}$ is defined as
$\rho(0) := -\infty$, $\rho(X_1^iX_2^j) := ir+j(r+1)$ and $\rho(\sum_i a_i X_1^{b_i}X_2^{c_i}) := \max_{i: a_i\ne 0} \rho(X_1^{b_i}X_2^{c_i})$.
For any $0\le m\le r^3+r^2-r-1$, the $q$-ary Hermitian codes are defined as
$\{(p(a_1),\dotsc,p(a_{r^3}))\mid p \in \mathbb{F}_q[X_1,X_2],\, \deg_1(p) < q,\, \deg_2(p) < r,\, \rho(p)\le m\}$,
where $\{a_1,\dotsc a_{r^3}\}$ is a set of zero points of $X_1^{q+1}-X_2^q-X_2$ in $\mathbb{F}_q$
and where $\deg_1(p)$ and $\deg_2(p)$ are the degree of $X_1$ of $p$ and the degree of $X_2$ of $p$, respectively.
\end{definition}

For a prime power $r$, 
a matrix $G_\text{H}(r^3)$ of size $r^3\times r^3$ can be defined for the $r^2$-ary Hermitian codes, 
just as the Reed-Solomon matrix $G_\text{RS}(q, \ell)$ has been defined 
for the $q$-ary Reed-Solomon codes.
In this paper, we call the matrix $G_\text{H}(r^3)$ the $r^2$-ary Hermitian matrix.
The Hermitian matrix $G_\text{H}(8)$ on $\mathbb{F}_4 = \{0,1,\alpha,\alpha^2\}$ is the following.
\begin{align*}
&(X_2,X_1):\\
&\hspace{3.3em}(\alpha^2,\alpha^2)\,(\alpha,\alpha^2)\,(\alpha^2,\alpha)\,(\alpha,\alpha)~(\alpha^2,1)~(\alpha,1)\;(1,0)\,(0,0)\\
&
\begin{array}{l}
\hspace{-0em}X_2X_1^3\\
\hspace{-0em}X_2X_1^2\\
\hspace{-0em}\phantom{X_2}X_1^3\\
\hspace{-0em}X_2X_1\\
\hspace{-0em}\phantom{X_2}X_1^2\\
\hspace{-0em}X_2\\
\hspace{-0em}\phantom{X_2}X_1\\
\hspace{-0em}\phantom{X_2}1\hspace{4em}
\end{array}
\hspace{-3em}
\begin{bmatrix}
\;~\alpha^2~\; & \;~\alpha~\; & \;~\alpha^2~\; & \;~\alpha~\; & \;~\alpha^2~\; & \;~\alpha~\; & \;0~ & ~0~\\
1 & \alpha^2 & \alpha & 1 & \alpha^2 & \alpha & 0 & 0\\
1 & 1 & 1 & 1 & 1 & 1 & 0 & 0\\
\alpha & 1 & 1 & \alpha^2 & \alpha^2 & \alpha & 0 & 0\\
\alpha & \alpha & \alpha^2 & \alpha^2 & 1 & 1 & 0 & 0\\
\alpha^2 & \alpha & \alpha^2 & \alpha & \alpha^2 & \alpha & 1 & 0\\
\alpha^2 & \alpha^2 & \alpha & \alpha & 1 & 1 & 0 & 0\\
1 & 1 & 1 & 1 & 1 & 1 & 1 & 1\\
\end{bmatrix}.
\end{align*}
From~\cite{yang-true}, the minimum distance of 4-ary Hermitian codes can be calculated as
$(D_0,\dotsc,D_{\ell-1})=(1,2,2,3,4,5,6,8)$.
Note that each $D_i$ is the largest under given blocklength and dimension
since the MDS conjecture has been proved for $q=4$~\cite{macwilliams1988tec}.
However, $E(G_\text{H}(8)) = L(4,8) \approx 0.562\,161 < 0.573\,12 \approx L(4,4)$.
Values of $E(G_\text{RS}(2^{m},2^{m}))$ and $E(G_\text{H}(2^{3m/2}))$ are shown in Table~\ref{table:e} for $m=2,4,6,8$.
The exponent of the Hermitian matrix $G_\text{H}(2^{3m/2})$ exceeds 
the exponent of the Reed-Solomon matrix $G_\text{RS}(2^{m},2^{m})$ for $m\ge 4$.
The method of shortening~\cite{korada2009pcc} may be useful for obtaining smaller matrices from Hermitian matrices,
although the exponent may also be smaller.

\begin{table}[tb]
\renewcommand{\arraystretch}{1.3}
\caption{
Exponents of Reed-Solomon kernels and Hermitian Kernels on $\mathbb{F}_{2^m}$.
}
\label{table:e}
\begin{tabular}{|l|l|l|l|l|}
\hline
$m$ & 2 & 4 & 6 & 8\\ 
\hline
$E(G_\text{RS}(2^m,2^m))$ & 0.573\,120 & 0.691\,408 & 0.770\,821 & 0.822\,264\\ 
\hline
$E(G_\text{H}(2^{3m/2}))$  & 0.562\,161 & 0.707\,337 & 0.802\,760 & 0.859\,299\\ 
\hline
\end{tabular}
\end{table}

\section{Conclusion}
We have shown error probabilities of $q$-ary polar codes using Reed-Solomon matrices as kernels by numerical simulations.
It is confirmed by numerical simulations
that 4-ary polar codes using Reed-Solomon matrix have significantly better performance than binary polar codes using Reed-Solomon matrix.
We have further shown that kernels with larger exponents can be obtained by using Hermitian codes.
This implies that algebraic geometry codes might be useful as kernels of polar codes with large exponents.

\section*{Acknowledgment}
Support from the Grant-in-Aid for Scientific Research (C), the Japan Society for the Promotion of Science, Japan
(No.~22560375) is acknowledged.
RM acknowledges support of Grant-in-Aid for JSPS Fellows (No.~22$\cdot$5936).

\bibliographystyle{IEEEtran}
\bibliography{IEEEabrv,ldpc}

\end{document}